\documentclass[]{article}

\usepackage[utf8]{inputenc}
\usepackage{microtype}
\usepackage{libertine}
\usepackage{libertinust1math}
\usepackage[T1]{fontenc}
\usepackage{amssymb}
\usepackage{amsfonts}
\usepackage{amsmath}
\usepackage{amsthm}

\usepackage{natbib}

\usepackage[english]{babel}

\RequirePackage[l2tabu, orthodox]{nag}
\usepackage[all,warning]{onlyamsmath}

\usepackage{hyperref}
\usepackage{cleveref}
\usepackage{color}
\usepackage{mathtools}
\usepackage{xspace}

\theoremstyle{plain}
\newtheorem{theorem}{Theorem}[section]
\newtheorem{lemma}[theorem]{Lemma}

\newtheorem{proposition}[theorem]{Proposition}

\theoremstyle{definition}
\newtheorem{definition}[theorem]{Definition}
\newtheorem{example}[theorem]{Example}

\newcommand{\set}[1]{\{#1\}}
\newcommand{\st}{\ :\ }

\newcommand{\torus}{\mathbb{T}}
\newcommand{\nat}{\mathbb{N}}
\newcommand{\intg}{\mathbb{Z}}
\newcommand{\rel}{\mathbb{R}}
\newcommand{\rat}{\mathbb{Q}}
\newcommand{\com}{\mathbb{C}}

\newcommand{\ii}{\mathbf{i}}

\newcommand{\sgn}{\mathrm{sgn}}

\newcommand{\lrs}{{\sc lrs}\xspace}
\newcommand{\mso}{{\sc mso}\xspace}
\newcommand{\ltl}{{\sc ltl}\xspace}

\newcommand{\lu}{\mathbf{u}}
\newcommand{\lv}{\mathbf{v}}
\newcommand{\sq}[1]{\langle #1 \rangle_{n\in\nat}}

\newcommand{\ie}{{\em i.e.}\xspace}
\newcommand{\eg}{{\em e.g.}\xspace}
\newcommand{\etc}{{\em etc.}\xspace}

\newcommand\defeq{\mathrel{\overset{\makebox[0pt]{\mbox{\normalfont\tiny\sffamily def}}}{=}}}

\newcommand{\aut}{\mathcal{A}}
\newcommand{\pow}{\mathcal{P}}

\newcommand{\inter}{\mathrm{inter}}

\newcommand{\red}[1]{{\color{red} #1}}

\newcommand{\mon}{\mathfrak{M}}

\newcommand{\mx}{\bar{x}}
\newcommand{\my}{\bar{y}}

\usepackage{thmtools}

\declaretheorem[name=Remark,numberwithin=section,style=remark]{remark}

\usepackage[auth-sc]{authblk}

\author[1]{Shaull Almagor}
\author[2]{Toghrul Karimov}
\author[3]{Edon Kelmendi}
\author[2,3]{J\"oel~Ouaknine}
\author[3]{James Worrell}
\affil[1]{Technion, Israel}
\affil[2]{Max Planck Institute for Software Systems, Germany}
\affil[3]{Oxford University, UK}
\title{Deciding $\omega$-Regular Properties on Linear Recurrence Sequences}
\begin{document}
\maketitle

\begin{abstract}
  We consider the problem of deciding $\omega$-regular properties on
  infinite traces produced by linear loops.  Here we think of a given
  loop as producing a single infinite trace that encodes information
  about the signs of program variables at each time step.  Formally,
  our main result is a procedure that inputs a prefix-independent
  $\omega$-regular property and a sequence of numbers satisfying a
  linear recurrence, and determines whether the sign description of
  the sequence (obtained by replacing each positive entry with
  ``$+$'', each negative entry with ``$-$'', and each zero entry with
  ``$0$'') satisfies the given property.  Our procedure requires that
  the recurrence be simple, \ie, that the update matrix of the
  underlying loop be diagonalisable.  This assumption is instrumental
  in proving our key technical lemma: namely that the sign description
  of a simple linear recurrence sequence is almost periodic in the
  sense of Muchnik, Sem\"enov, and Ushakov.  To complement this lemma, we
  give an example of a linear recurrence sequence whose sign
  description fails to be almost periodic. Generalising from sign
  descriptions, we also consider the verification of properties involving
  semi-algebraic predicates on program variables.
\end{abstract}




\maketitle

\section{Introduction}
The decidability of monadic second-order logic (\mso) over the structure $\langle \mathbb{N},< \rangle$ is a pillar of the theory of automated verification~\cite{Buchi62}.  Shortly after B\"{u}chi established this result, \citet{elgot66_decid_undec_exten_secon_order_theor_succes} began to investigate unary predicates $\boldsymbol{P} \subseteq \mathbb{N}$ for which the \mso theory of the structure $\langle \mathbb{N},<,\boldsymbol{P} \rangle$ remains decidable.  For example, decidability is known in case $\boldsymbol{P}$ denotes, respectively, the set of factorial numbers $\set{n!:n\in\nat}$, the set $\set{n^k:n\in\nat}$ for a fixed but arbitrary natural $k$, and the set of $k$-powers $\{k^n:n\in\nat\}$ for every fixed $k\in\mathbb{N}$.  On the other hand, there are natural examples of predicates for which decidability is open and apparently difficult, e.g., for $\boldsymbol{P}$ the set of primes, see~\cite{Bateman1993}.\footnote{Note that the twin primes conjecture from number theory (that there are infinitely many pairs of primes that differ by two) can be formulated in \mso with the primality predicate.}  In general, the decision problem for \mso over $\langle\mathbb{N},<,\boldsymbol{P}\rangle$ reduces to the problem of checking membership of a fixed $\omega$-word (namely the characteristic word of the predicate $\boldsymbol{P}$) in an $\omega$-regular language $\mathcal{L}$ that represents the formula whose truth is to be determined.

\citet{semenov84_logic_theor_one_place_funct} gave a characterisation of those predicates $\boldsymbol{P}$ for which the \mso theory of $\langle \mathbb{N},<,\boldsymbol{P} \rangle$ is decidable.  This line of work was continued in~\cite{carton02_monad_theor_morph_infin_words_gener,rabinovich07_decid_monad_logic_order_over}.
An important sufficient condition for decidability is that the
characteristic sequence of $\boldsymbol{P}$ be effectively \emph{almost periodic}:
a notion originating in symbolic
dynamics~\cite{morse38_symbol_dynam,muchnik03_almos_period_sequen}.
Roughly speaking, a sequence is almost periodic if for any pattern
that occurs infinitely often, the gaps between successive occurrences
are bounded.  A classical example of an almost periodic sequence is
the Thue-Morse sequence, \emph{i.e.}, the sequence whose $n$-th entry is the
parity of the number $1$'s in the binary expansion of $n$.  Another
example is the characteristic sequence of the predicate
$\boldsymbol{P}=\{ n \in \mathbb{N}:\sin(n\theta)>0\}$ for a fixed
real number $\theta$.  The notion of almost periodicity will be
instrumental for the results of this paper.

From the point of view of program analysis, it is natural to consider
extensions of \mso with unary predicates that encode properties of program
variables at each time step.  For example, consider a linear loop 
\[ \mathbf{while}\; \mathrm{true}\; \mathbf{do} \; 
\begin{pmatrix} x\\y\\z \end{pmatrix} \leftarrow
\begin{pmatrix} -2&4&0\\
                 4&0&0\\
                 1&0&0
\end{pmatrix}
\begin{pmatrix} x\\y\\z \end{pmatrix} \, . \]
Given initial values of the program variables, suppose we want to
determine whether the variable $x$ is ultimately increasing.  Noting
that variable $z$ stores the previous value of $x$, we equivalently
want to determine ultimate positivity of the sequence $\langle u_n
\rangle_{n\in\mathbb{N}}$ defined by $u_n:=x_n-z_n$, where $x_n$ and
$z_n$ are the respective values of variables $x$ and $z$ after $n$
executions of the loop body.  This property can be written in \mso as
$\exists m \, \forall n \cdot (n\ge m \Rightarrow \boldsymbol{P}(n))$
where $\boldsymbol{P} = \set{ n \in \mathbb{N} : u_n > 0}$.
Clearly we can use second-order quantification in \mso to express more
complex properties, e.g., that variable $x$ only increases on even
steps of the execution.

In the above example, the sequence $\lu=\langle u_n
\rangle_{n\in\mathbb{N}}$ satisfies the recurrence
$u_{n+2}=-2u_{n+1}+16u_n$.  In general, for any linear loop and any
polynomial function on the program variables, the sequence of values
assumed by the function along an infinite execution of the loop is a
linear recurrence sequence (\lrs).  This observation motivates the
central object of study in this paper: the decidability of the \mso
theory of the structure $\mathcal{S}_\lu:=\langle
\mathbb{N},\leq,\boldsymbol{P},\boldsymbol{Z},\boldsymbol{N}\rangle$
associated with an \lrs $\lu$, where
$\boldsymbol{P}:=\{n\in\mathbb{N}:u_n>0\}$,
$\boldsymbol{Z}:=\{n\in\mathbb{N}:u_n=0\}$, and
$\boldsymbol{N}:=\{n\in\mathbb{N}:u_n<0\}$.  This structure can be
represented by the \emph{sign description} ${\langle
  \mathrm{sgn}(u_n)\rangle_{n \in \mathbb{N}}} \in {\{+,0,-\}^\omega}$
of $\lu$, which is defined in the obvious way.

Computational problems concerning sign descriptions of \lrs are
notoriously difficult.  For example, decidability of the Skolem
Problem \emph{``does a given \lrs have a zero term?''} has been open
for many decades.  Decidability of the Positivity Problem: \emph{``are
  all terms of a given \lrs positive?''} is likewise a longstanding
open
problem~\cite{ouaknine13_posit_probl_low_order_linear_recur_sequen,SalomaaS78}.
In view of these difficulties, we restrict attention to the class of
\emph{simple} \lrs, \emph{i.e.}, those such that the characteristic
polynomial of the defining recurrence has simple roots.  In terms of
our motivating example of linear loops, the associated \lrs are simple
whenever the update matrix of the loop is diagonalisable.  Our main
technical lemma shows that the sign description of every simple \lrs
is (effectively) almost periodic.  We moreover give an example showing
that almost periodicity fails without the assumption of simplicity.

Using the fact that simple \lrs have effectively almost periodic sign
descriptions, we establish our first main result:
\begin{restatable}{thm}{secondmainth}
  \label{th:secondmain}
For every fixed simple \lrs $\lu$ of rational numbers, it is decidable
whether the sign description of $\lu$ lies in a given $\omega$-regular
language $\mathcal L$.
\end{restatable}
This result yields a large new class of structures with a decidable \mso
theory, namely each structure $\mathcal{S}_\lu$ for $\lu$ a simple
\lrs.  

We emphasize that Theorem~\ref{th:secondmain} states the existence of
a decision procedure for every fixed \lrs $\lu$.  For our application
of model checking linear loops, it is more natural to consider the \lrs
$\lu$ as part of the input to the decision procedure, since the
definition of $\lu$ depends on the loop.  This leads to our second
main result:
\begin{restatable}{thm}{mainth}
  \label{th:main}
  Given a prefix-independent $\omega$-regular language $\mathcal L$
  and a simple \lrs $\lu$, it is decidable whether the sign
  description of $\lu$ belongs to $\mathcal L$.
\end{restatable}
This result allows us to model check prefix-independent \mso
properties that refer to the signs of variables in linear loops.
Recall here that a prefix-independent $\omega$-regular language is one
such that any two words with a common (infinite) suffix are either
both in the language or both not in the language.  Equivalently, such
a language is a finite union of languages of the form $\Sigma^*
\mathcal{L}^\omega$ for regular $\mathcal{L}\subseteq \Sigma^*$.
Intuitively, prefix-independent languages specify asymptotic
properties of $\omega$-words, such as that a certain pattern occurs
infinitely often or that some property eventually holds forever.  The
restriction to prefix-independent properties in~\Cref{th:main} is
connected to a non-uniformity in the proof that the sign description
of a simple \lrs is effectively almost periodic
(cf.~\Cref{thm:simple almost periodic} and~\Cref{rem:n_0}), which in
turn is due to our use of ineffective number-theoretic bounds.  Note
that the ability to handle arbitrary $\omega$-regular languages in
Theorem~\ref{th:main} would immediately entail decidability of both
Skolem's Problem and the Positivity Problem for simple \lrs.

The sign description is a coarse abstraction of a given
sequence. However the same techniques that provide us with
\Cref{th:main} can be applied to substantially more powerful
abstractions, as we illustrate below once more in the context of analysing the
behaviour of a linear loop. Let us assume that such a loop operates over
variables $x_1, \ldots, x_m$; writing $\mathbf{v}$ for the column
vector of variables $[x_1, \ldots, x_m]^T$, we represent our loop by
the expression
\[ \mathbf{while}\; \mathrm{true}\; \mathbf{do} \; 
\mathbf{v} \leftarrow M \mathbf{v} \, , \]
where $M$ is a square matrix of dimension $m \times m$. (Note that we
have written the guard as `$\mathrm{true}$' since we are only interested
in the properties of \emph{infinite} executions of loops.)
Let us now assume that we are given $k$ semi-algebraic predicates over
variables $x_1,\ldots,x_m$, \emph{i.e.}, Boolean combinations of polynomial
inequalities on these variables, and let us denote these predicates as
$\boldsymbol{P}_1, \ldots, \boldsymbol{P}_k$; we shall naturally
identify such predicates with the semi-algebraic subsets of
$\mathbb{R}^m$ that they represent. Now given any $m$ initial
values $a_1, \ldots, a_m$ for the variables $x_1, \ldots, x_m$,
executing our loop gives rise to an infinite sequence of points (or \emph{orbit}) 
$\langle \mathbf{v_0}, \mathbf{v_1}, \ldots \rangle$ in
$\mathbb{R}^m$: we have $\mathbf{v_0} =  [a_1, \ldots, a_m]^T$, and
for all $i \geq 0$, $\mathbf{v_{i+1}} = M \mathbf{v_i}$. These are, of
course, the infinite sequence of values that the loop variables take
as the loop forever unwinds over time. Note that predicate $\boldsymbol{P}_i$
is satisfied at time $n$ iff $\mathbf{v_n} \in \boldsymbol{P}_i$.

Finally, let us assume that we are given a prefix-independent \mso
property over predicates $\boldsymbol{P}_1, \ldots, \boldsymbol{P}_k$,
describing some specification that the infinite unwinding of the loop,
given the initial assignment of values to the variables, may or may
not satisfy. Theorem~\ref{th:extension2} asserts that, provided
that the matrix $M$ is diagonalisable, the model-checking problem of
whether the orbit of the loop satisfies the given \mso property is
decidable. As for the sign description of \lrs, the proof technique relies on
almost periodicity and Theorem~\ref{th:extension2} is proved in a
manner similar to~\Cref{th:main}; see~\Cref{subsec:sa predicates}.


\paragraph{Related Work}
There have been a number of previous works that introduce symbolic
semantics for linear systems, including linear loops and Markov chains,
and give model checking procedures for this semantics.  But the
current paper is the first that establishes and benefits from almost
periodicity of a symbolic semantics.

The paper~\cite{ltlpaper} examines a version of the decision problem
considered in this paper, but with \ltl formulas rather than \mso
formulas, and restricting to recurrences of order at most~3
(corresponding to linear loops with at most 3 variables).  In addition
to the restriction on order, a major difference with the present paper
is that~\cite{ltlpaper} does not use the notion of almost periodic
sequences.  Intuitively the model checking problem can be handled more
directly there by exploiting the simplicty of \ltl.

The paper~\cite{BeauquierRS06} considers \mso over $\langle
\mathbb{N},<\rangle$ augmented with a probability quantifier.  The
semantics of the probability quantifier is defined relative to
trajectories of a finite-state Markov chain.  The setting is close to
the present paper: Markov chains are a special case of linear loops,
and the probability quantifier corresponds to having predicates that
report the sign of an \lrs at each index.  However the results
of~\cite{BeauquierRS06} only apply in situations in which truth values
of formulas are ultimately periodic.  By working with the notion of
almost periodicity we avoid the need for such semantic restrictions.
Interestingly,~\cite[Section 8]{BeauquierRS06} notes the close
relationship to the model checking problem for their logic and the
Skolem Problem for linear recurrences.

Another similar work is~\cite{AgrawalAGT15}, which considers the
problem of model checking \ltl formulas on a symbolic dynamics of a
Markov chain that is induced by a finite polyhedral partition of the
space of probability distributions on the states.  Again, the key
issue is ultimate periodicity: the authors of~\cite{AgrawalAGT15} note
that their symbolic dynamics is not ultimately periodic in general,
and therefore switch their attention of a notion of approximate model
checking.
\newpage

Decision problems on the positivity of \lrs have been studied
in~\cite{OuaknineW14a,OuaknineW14b,ouaknine13_posit_probl_low_order_linear_recur_sequen}.
Our second main result, Theorem~\ref{th:main}, generalises the fact
that it is decidable whether a simple \lrs is ultimately
positive~\cite{OuaknineW14a}.  In terms of the structure of the sign
description of an \lrs,~\cite{BG07} show that the positivity set of
an \lrs (the set of indices where the \lrs is positive) has a density
and characterises the numbers that can appear as such a density.  A
classical result of Skolem, Mahler, Lech states that the set of zeros
of an \lrs over a field of characteristic zero is ultimately periodic.

\paragraph{Organisation} 
The rest of the paper is organised as follows. In~\Cref{sec:desc}, we give the main definitions and discuss two classical results: the Skolem-Mahler-Lech theorem and Sem\"enov's theorem. We sketch the plan for the proof of the main theorem in~\Cref{subsec:sketch}. The central technical theorem is proved in \Cref{sec:ap}. In the last subsection, properties related to effectiveness of the objects defined in the proof are given. In~\Cref{sec:cexample} we show that the sign descriptions of general \lrs need not be almost periodic. This section is independent and can be read out of order. In~\Cref{sec:omegareg} we give the procedure and in~\Cref{sec:extensions} we show how the proof can be adapted to more complex predicates instead of sign descriptions. The pertinent notions of the first-order theory of real closed fields and related proofs are presented in~\Cref{sec:foth}. 


\section{Sign descriptions of linear recurrence sequences}
\label{sec:desc}
A \emph{linear recurrence sequence (\lrs)} is a sequence $\lu=\sq u$ of rational numbers that satisfies
a recurrence relation
\begin{align}
  u_n=a_1u_{n-1}+a_2u_{n-2}+\cdots+a_du_{n-d},\qquad n>d,
\label{eq:recur-relation}
\end{align}
where $a_1,\ldots,a_d$ are rational constants and $d\in\mathbb{N}$ is the order of recurrence.
Clearly such a sequence is determined by the recurrence and the initial values $u_1,\ldots,u_d$.

The \emph{characteristic polynomial} of the recurrence~\eqref{eq:recur-relation} is
\begin{align*}
  f(x)\defeq x^{d}-a_1x^{d-1}-\cdots-a_{d-1}x-a_d.
\end{align*}
We refer to the roots of $f$ as the \emph{characteristic roots} of the
recurrence.  It is well known that an \lrs $\lu$ admits a unique representation as
an \emph{exponential polynomial}
\[ u_n = \sum_{i=1}^m C_i(n) \Lambda_i^n,\] where $\Lambda_1,\ldots,\Lambda_m$ are the distinct characteristic roots and the $C_i$ are polynomials. Both the roots and the coefficients of the polynomials $C_i$ are in general complex algebraic numbers.

An \lrs satisfies a unique recurrence of minimum order.  We say that
the recurrence is \emph{simple} if the characteristic roots of this
recurrence are simple.  Equivalently $\lu$ is simple if the coefficients
$C_i$ in its representation as an exponential polynomial are
constant polynomials.


Let $\lu=\sq u$ be a linear recurrence sequence. Define $\zeta$, an infinite word over the alphabet~$\set{0,\pm}$, as:
\begin{align*}
  \zeta_n \defeq 0\qquad \Leftrightarrow\qquad u_n=0. 
\end{align*}
In other words, we abstract away the terms of the sequence and only
keep the information of whether or not they are equal to zero. The
celebrated Skolem-Mahler-Lech theorem says that the word $\zeta$ is
ultimately periodic.
\begin{theorem}[Skolem-Mahler-Lech, {\cite[Theorem 2.1]{recseq}}]
  \label{th:skolem}
  For any linear recurrence sequence $\lu$ the word $\zeta$ is of the form
  \begin{align*}
    \zeta = w_1w_2^\omega,
  \end{align*}
  for $w_1,w_2\in\set{0,\pm}^*$. 
\end{theorem}
The word $w_2$ can be computed \cite{berstel76_deux_des_suites} from
the description of $\lu$; it is however a longstanding open problem
whether the same is true for the prefix $w_1$.

In this paper we are interested in a slightly finer analysis:
\begin{definition}[Sign description]
The sign description of $\lu$ is the infinite word $\sigma$ over the alphabet~$\set{-,0,+}$ defined as:
\begin{align*}
  \sigma_n \defeq \sgn(u_n),
\end{align*}
where for $x\in\rel$,
\begin{align*}
  \sgn(x)\defeq\begin{cases}
    &+ \qquad \text{if }x>0,\\
    &- \qquad \text{if }x<0,\\
    &0 \qquad \text{otherwise}.
    \end{cases}    
\end{align*}
\end{definition}
Unlike $\zeta$, the word $\sigma$ is not ultimately periodic in general, as the following example shows.
\begin{example}
  Let $\lu$ be an \lrs given in the matrix form\footnote{This is an equivalent formulation for \lrs, inter-reducible in polynomial time with the definition that we gave in the beginning of this section; see {\cite[Section 1.1.12]{recseq}}.} as:
  \begin{align*}
    u_n=\begin{pmatrix} 0 & 1 \end{pmatrix}
        \begin{pmatrix}a & b\\ -b & a\end{pmatrix}^n
        \begin{pmatrix} 0 \\ 1 \end{pmatrix},\qquad a,b\neq 0, \text{ and }a^2+b^2=1. 
  \end{align*}
  Putting the square matrix above in Jordan normal form and using Euler's formula, we can deduce that $u_n=\cos(n\varphi)$, where $\varphi=\arg(a+\ii b)$. The set
  \begin{align*}
    \set{a+\ii b \st a,b\in\rat,\ a,b\ne 0\text{ and } a^2+b^2=1},
  \end{align*}
  consists of algebraic numbers of degree two. The only roots of unity of degree two are the third, fourth and sixth primitive roots of unity, which are either $\pm\ii$ or have irrational imaginary part. Consequently, none of the elements of the set above are a root of unity and therefore, $\varphi=\arg(a+\ii b)$ is not a rational multiple of $\pi$.

  If the sign description $\sigma$ were to be ultimately periodic, then there would be some $p\in\nat$ and $s\in\set{-,0,+}$ such that $\sigma_{np}=s$ for all $n\in\nat$. This is not the case, because $p\varphi$ is not a rational multiple of $\pi$, from which it easily follows (e.g., using Kronecker's theorem for inhomogeneous Diophantine approximation,~\Cref{th:kronecker}) that $\set{\cos\left(n(p\varphi)\right)\st n\in\nat}$ is dense in $[-1,1]$.
\end{example}

However, in the case of simple \lrs, the sign description is well-behaved. In the sequel we will prove that simple \lrs have {\em almost periodic} sign descriptions. 
\subsection{Almost periodic words}
We say that the pattern $w\in\Sigma^*$ occurs in a word $\alpha \in \Sigma^*\cup \Sigma^\omega$
if $w$ occurs as an infix of $\alpha$.  More specifically, we say that $w$ occurs 
in position $p$ of $\alpha$ if 
\begin{align*}
  w=\alpha_p\alpha_{p+1}\cdots\alpha_{p+|w|-1}. 
\end{align*}

\begin{definition}[Almost periodic]
  \label{def:ap}
  An infinite word $\alpha\in\Sigma^\omega$ is almost periodic if for every word $w\in\Sigma^*$, there exists $p\in\nat$ such that either:
  \begin{itemize}
  \item $w$ does not occur in $\alpha$ after the position $p$, or
  \item $w$ occurs in every factor of $\alpha$ of length $p$, \ie for every $n\in\nat$, $w$ occurs in
    \begin{align*}
      \alpha_{n}\alpha_{n+1}\cdots\alpha_{n+p}.
    \end{align*}
  \end{itemize}
\end{definition}
Intuitively, an almost periodic word is one with the property that any pattern that occurs infinitely often, does so in such a manner that
 the gaps between successive ocurrences of the pattern have bounded length.
A typical {\em non-example} of almost periodic words is:
\begin{align*}
  aba^2ba^3ba^4b\cdots.
\end{align*}
Here the letter $b$ occurs infinitely often, but the distances between consecutive occurrences are unbounded. 

Almost periodic words are sometimes referred to in the literature as
\emph{uniformly recurrent sequences}, or \emph{minimal sequences}. As
examples of almost periodic words we have: ultimately periodic words,
Sturmian words, and some morphic sequences such as the Thue-Morse
sequence. Almost-periodic words enjoy good closure properties, low
Kolmogorov complexity, \etc; \cite{muchnik03_almos_period_sequen} is
an extensive study on the combinatorics of these words.

An almost periodic word $\alpha \in \Sigma^\omega$
is said to be \emph{effectively
  almost periodic} if, given a pattern $w\in \Sigma^*$, we can decide whether or not 
$w$ occurs infinitely often in $\alpha$, and, if so, we can 
compute an upper bound $p$ between successive occurrences of $w$ in $\alpha$. If the pattern does not occur infinitely often on the other hand, we can compute an upper bound on the threshold after which the pattern does not occur. Equivalently, $\alpha$ is effectively almost periodic if there is a procedure that inputs a pattern $w\in\Sigma^*$ and outputs an upper bound on the number $p$ in \Cref{def:ap}.

A key property of effectively almost periodic
words is that they have a decidable monadic second-order theory.  More
specifically, a word $\alpha \in \Sigma^\omega$ determines a structure that expands
$(\nat,<)$ with a monadic predicate for every letter in $\Sigma$
that denotes the positions in $\alpha$ where the letter occurs.  Formulas of
\mso over this structure are formulas of predicate logic with
both first-order variables and monadic
second-order variables.  Then we have:
\begin{theorem}[{\cite[Theorem 1]{semenov84_logic_theor_one_place_funct}}]
  \label{thm:semenov}
  For any effectively almost periodic word $\alpha$, the \mso theory
  of $(\nat,<)$ expanded with unary predicates that define $\alpha$,
  is decidable.
\end{theorem}


One of the main results of this paper is that the sign description of a given simple \lrs is an effectively almost periodic word.  This effectiveness, however, is non-uniform, in the sense that we do not have a single algorithm that takes an \lrs as input and witnesses the effectiveness of the corresponding sign description.  Indeed such a uniform effectiveness result would allow to decide Skolem's Problem (``Does an \lrs have a zero term?'') and the Positivity Problem (``Are all terms of an \lrs positive?''), both of which are open for simple \lrs.  This fact leads us to formulate and prove a variant of \Cref{thm:semenov} that assumes a weaker notion of effectiveness that talks only about the asymptotic properties of the word.  Specifically this notion asks to compute an upper bound on the gap between all but finitely many succcessive occurrences of an infinitely recurring factor.  Naturally, for such sequences we correspondingly weaken the conclusion of \Cref{thm:semenov}: we ask to decide any {\em prefix-independent} $\omega$-regular property of the sign descriptions.


\subsection{Proof Sketch}
\label{subsec:sketch}

The proof of the main theorem, \Cref{th:main}, can be conceptually divided as follows: (\textbf{a}) we observe that simple linear recurrence sequences admit almost periodic sign descriptions, (\textbf{b}) we prove that there is a procedure that given a pattern, outputs a bound on distances between consecutive occurrences (in the sign description); and finally exploiting the previous procedure we provide (\textbf{c})~an algorithm that inputs a prefix-independent $\omega$-regular language $\mathcal L$ (as a M\"uller automaton)  and a simple \lrs, and decides whether its sign description belongs to $\mathcal L$. 

(\textbf{a}) To show almost periodicity, the general idea is to construct a much simpler dynamical system and  prove that its sign description coincides with that of the given sequence in all but finitely many positions. The ambient space of this dynamical system (described in \Cref{subsec:walks}) is a compact subset $X$ of $\torus^d$ (where $\torus$ is the unit circle on the complex plane). Its dynamics is given by a continuous function mapping $X$ to itself. This system is easier to analyse: for every sign pattern there exists an open subset $Y$ of $X$ such that when the system enters it, the next signs that it outputs form the pattern. Furthermore, using the compactness of $X$ we can prove that from everywhere in $X$, the system has to enter $Y$ in a bounded number of steps (provided $Y$ is non-empty). This bound will suffice for the distance between consecutive occurrences of the pattern.

  The reason why the sign sequences of the given system and the simpler one above coincide, in all but finitely many positions, is laid in \Cref{subsec:s-unit}. It amounts to proving that the asymptotic behavior of the sequence is determined by its dominant terms (those made from characteristic roots with maximal modulus). To lower bound these terms, we will apply a theorem from algebraic number theory.
  
(\textbf{b}) The procedure for calculating the distances between occurrences of patterns, manipulates formulas of  first order logic of the field of real numbers. We observe that for every pattern, the subsets $Y$ above, are semi-algebraic (\ie they are definable in the logic), and that furthermore the formulas can be effectively computed (\Cref{lem:fodef}). Using Tarski's procedure we can check whether $Y$ is non-empty, \ie whether the pattern occurs infinitely often in the sign description, and if so, calculate the bound between consecutive occurrences by querying whether the bound $b$ is sufficient, for successive $b\in\nat$~(\Cref{prop:properties}).

(\textbf{c}) We gather all the relevant properties of the sign description in~\Cref{prop:properties}, abstracting away linear recurrence sequences; so that the algorithm that is presented in \Cref{sec:omegareg} would work for any infinite word having the properties listed in~\Cref{prop:properties}.

Because it is simpler for the proofs, the algorithm will manipulate elements of a certain finite monoid which is equivalent to the given automaton. The sign description $\sigma$ has the property that one can choose finite words $w_1,\ldots, w_k$ such that
--- except for a finite prefix --- $\sigma$ is obtained by
intercalating the words $w_1,\ldots,w_k$.  As a consequence of the automaton being finite, for some well chosen and sufficiently long words $w_1,\ldots, w_k$, we can prove that it does not matter for the acceptance how they are arranged in the suffix of $\sigma$. The algorithm will construct these sufficiently long words and multiply the associated elements of the monoid to decide whether the set of states that is seen infinitely often is final. 


\section{Simple \lrs have almost periodic sign descriptions}
\label{sec:ap}
In this section we prove our first main result:
\begin{theorem}
  \label{thm:simple almost periodic}
  The sign description of a simple linear recurrence sequence is almost periodic.
\end{theorem}

Fix a simple \lrs $\lu$.  We first give a brief informal overview of the
proof.\footnote{In fact the technical details, below, will depart
  slightly from this overview due to the need to handle the issue of
  degeneracy of \lrs.}  To set up the idea of the proof, recall that
$\lu$ admits a representation as an exponential polynomial
\begin{align}
  u_n=\sum_{i=1}^d c_i\ \Lambda_i^n,
\label{eq:exp-poly}
\end{align}
where $c_i,\Lambda_i$ are non-zero algebraic numbers, with $\Lambda_i$ being characteristic roots of the recurrence defining $\lu$.  Now for each $i\in\{1,\ldots,d\}$, we factor each $\Lambda_i$ as the product $\Lambda_i=\rho_i\lambda_i$ of a positive real number $\rho_i>0$ and a complex number $\lambda_i$ of absolute value $1$.  The first key idea is that for $n$ sufficiently large, the sign of $u_n$ is determined by $(\lambda_1^n,\ldots,\lambda_d^n)$, i.e., the absolute values of the characteristic roots can be ignored for large $n$.  The second key idea is that the set $\{(\lambda_1^n,\ldots,\lambda_d^n)\st n\in\mathbb{N}\}$ is the orbit of a point under a homeomorphism of a compact topological space, namely the $d$-fold product of the unit circle $\mathbb{T}$ in the complex plane.  This transports us to a classical situation in symbolic dynamics.

As a preliminary step, we first decompose $\lu$ as the interleaving 
of several so-called \emph{non-degenerate} subsequences.
Recall here that an \lrs is said to be non-degenerate if no quotient
of two distinct characteristic roots is a root of unity.  To decompose
$\lu$, as given in~\eqref{eq:exp-poly},
we take $P\in\nat$ to be the least common multiple
of the orders of all roots of unity among the quotients $\Lambda_i/\Lambda_j$ for $1\leq i<j\leq d$;
then for all $\ell\in\nat$, $0\le \ell < P$, the sequence
\begin{align*}
  \lu^{(\ell)} \defeq \sq{u_{\ell+nP}},
\end{align*}
is a non-degenerate \lrs with characteristic roots among $\{\Lambda_1^P,\ldots,\Lambda_d^P\}$.
We factor the characteristic roots as 
\begin{align}
  \label{eq:normalized roots}
  \rho_i\lambda_i\defeq\Lambda_i^P\qquad \rho_i\in\rel_+,\  |\lambda_i|=1,\ 1\le i\le d.
\end{align}

The rationale behind this decomposition is that non-degenerate
sequences have the following property:
\begin{proposition}[{\cite[Corollary 2.1]{shapiro59_theor_concer_expon_polyn}}]
  \label{prop:shapiro}
  A non-degenerate \lrs either has finitely many zeros, or it is identically zero.
\end{proposition}

Next we will demonstrate that the sign description of $\lu^{(\ell)}$
is asymptotically the same as that of a certain linear function on
$(\lambda_1^n,\ldots,\lambda_d^n)$, \ie, it does not depend on the
moduli $\rho_i$. We achieve this by applying the work of Evertse, van
der Poorten, and Schlickewei on bounds of sums of S-units.

\subsection{A lower bound on sums of $S$-units}
\label{subsec:s-unit}

We will prove the following lemma.

\begin{lemma}
  \label{lem:same sign}
  Let $\lv\defeq\lu^{(\ell)}$, for some $0\le\ell<P$. There exist
  $z_1,\ldots,z_d\in\com$ such that $\sum_{i=1}^d z_i \lambda_i^n$ is
  real for all $n\in\nat$, and furthermore exists $n_0\in\nat$ such
  that for all $n\ge n_0$,
  \begin{align*}
    \sgn(v_n)=\sgn\left(\ \sum_{i=1}^dz_i\lambda_i^n\ \right).
  \end{align*}
\end{lemma}
\begin{remark}
  \label{rem:n_0}
  There is no known effective means of determining the constant $n_0$
  above. For such a method we would need an effective version of
  Roth's theorem (consult Section 2.4 in \cite{recseq}). It is as a
  consequence of the ineffectiveness of this constant that we are
  forced to restrict to \emph{prefix-independent} $\omega$-regular
  properties in the main theorem. It is worth noting, however, that
in the presence of at most three dominant roots, this constant is effective~\cite[Theorem 1]{84_distan_between_terms_algeb_recur_sequen}. 
\end{remark}

The principal ingredient in the proof of \Cref{lem:same sign} is the
aforementioned lower bound on sums of $S$-units. We introduce this theorem first. 

Let $K$ be the splitting field of the characteristic polynomial, that is the field extension of $\rat$ generated by the characteristic roots $\Lambda_1,\ldots,\Lambda_d$. The elements of $K$ that are roots of monic polynomials in $\intg[x]$ (i.e., with leading coefficient one) form a subring, known as the {\em algebraic integers} of $K$, denoted $\mathcal{O}_K$. Further, $\mathcal{O}_K$ is a Dedekind ring, so for every $x\in\mathcal{O}_K$, the principal ideal generated by $x$ can be written down as a product of a finite number of prime ideals. Let $S$ be a finite set of prime ideals. An $S$-unit is any $x\in\mathcal{O}_K$ such that the prime divisors of the principal ideal of $x$ are in $S$.

If $K$ has degree $r$ over $\mathbb{Q}$ then there are $r$
field embeddings from $K$ to $\com$, denoted $h_1,\ldots,h_r$.

\begin{theorem}[Evertse, van der Poorten and Schlickewei, see \eg {\cite[Theorem 2]{evertse}}]
  \label{th:evertse}
  Let $S$ be a finite set of prime ideals in $\mathcal{O}_K$, and $m\in\nat$. Then for all $\epsilon>0$ there exists
  $C>0$, depending on $\epsilon$ and $m$, such that for any set of $S$-units $x_1,\ldots,x_m\in\mathcal{O}_K$, with the property that $\sum_{i\in I}x_i\ne 0$, $I\subseteq\set{1,2,\ldots,m}$, we have
  \begin{align*}
    \left |\sum_{i=1}^mx_i\right |\ge CXY^{-\epsilon},
  \end{align*}
  where $X\defeq\max\set{|x_i|\st 1\le i\le m}$, and $Y\defeq\max\set{|h_j(x_i)|\st 1\le i\le m, 1\le j\le r}$. 
\end{theorem}

We show how we can apply Theorem~\ref{th:evertse} to our setting.

Let $\ell\in\nat$, $0\le \ell < P$, and $\lv=\lu^{(\ell)}$. Assume that $\lv$ is not identically zero, then there exists $J\subseteq\set{1,\ldots,d}$, and $b_j\in \com$, $j\in J$ with $b_j\ne 0$ such that
\begin{align*}
  v_n= \sum_{j\in J}b_j(\rho_j\lambda_j)^n. 
\end{align*}
Let $J'\subseteq J$ be the dominant roots (with modulus $\rho$) among the roots in $J$ and write
\begin{align*}
  v_n= \underbrace{\sum_{j\in J'}b_j(\rho\lambda_j)^n}_{D(n)} + \underbrace{\sum_{j\in J\setminus J'}b_j(\rho_j\lambda_j)^n}_{R(n)}. 
\end{align*}
We will use Theorem~\ref{th:evertse} to show that the sign of $\lv$ asymptotically depends only on that of $D(n)$, by noticing that $D(n)$ is a sum of S-units. 

For \lrs over integers, the roots of the characteristic polynomial, as well as the analogue of the constants $c_i$ are algebraic integers in the respective splitting field. We have defined $\lu$ over rationals, however this can be sidestepped by observing that there exist natural numbers $a$ and~$b$ such that the entries of the sequence\footnote{This is a linear recurrence sequence because the point-wise product of two \lrs is again a \lrs.}  $\sq{ab^n u_n}$ are all integers, and furthermore it has the same sign description as $\lu$. Consequently we can assume that numbers $c_i$ and $\Lambda_i$ in the exponential polynomial description~\eqref{eq:exp-poly} are algebraic integers in $K$. Since $\mathcal{O}_K$ is a ring, it follows that the terms of the sum $D(n)$ above all belong to $\mathcal{O}_K$. 


Define $S$ to be the set of prime divisors of $(\rho_j\lambda_j)$ and
prime divisors of $b_j$. By definition of $S$ all
$(\rho_j\lambda_j),b_j$ are $S$-units and consequently $D(n)$ is a sum
of $S$-units. To apply Theorem~\ref{th:evertse}, we need now only
show that any sub-sum of $D(n)$ vanishes for only finitely many
$n$. To see this, observe that any sub-sum of $D(n)$ is itself a
non-degenerate \lrs, moreover we have assumed that it is not
identically zero (because $b_j\ne 0$); as a consequence of
\Cref{prop:shapiro}, it cannot vanish for infinitely many~$n$.

We now apply Theorem~\ref{th:evertse} to the sum of $S$-units $D(n)$.
In this situation, for all but finitely many $n$, we clearly have $X=|b|\rho^n$ for some $b=b_j, j\in J'$. Since for every root $(\rho_j\lambda_j)$ there is a field embedding among $h_1,\ldots,h_r$ that fixes it, for all but finitely many $n$, we have $Y \ge |b'|\rho^n$, for some constant $b'$. It follows that for every 
$\epsilon>0$ there exists $C>0$ such that for all but finitely many
$n$, we have
\begin{align}
  \label{eq:bound}
  \left |D(n)\right | \ge C\rho^{n(1-\epsilon)}.
\end{align}

We are now ready to prove Lemma~\ref{lem:same sign}.
\begin{proof}[Proof of Lemma~\ref{lem:same sign}]
  If $\lv$ is identically zero the lemma clearly holds. Assume that $\lv$ is not identically zero. Since $\rho>\rho_j$ for $j\in J\setminus J'$, we have that there exists some $\epsilon_1>0$ such that for all but finitely many $n$,
\begin{align*}
  \left |R(n) \right |<\rho^{n(1-\epsilon_1)}.
\end{align*}
Since \eqref{eq:bound} holds for any $\epsilon>0$ it follows now that for all but finitely many $n$,
\begin{align*}
  \left | D(n) \right | > \left | R(n) \right |. 
\end{align*}
For all but fintely many $n$ we thus have 
\[ \sgn(v_n)=\sgn(D(n))=\sgn\left(\sum_{j\in J'}b_j\lambda_j^n\right)\, .\]
This completes the proof the lemma.
\end{proof}
\subsection{Orbits in $\torus^d$}
\label{subsec:walks}
Lemma~\ref{lem:same sign} tells us that the information about the sign
description $\sigma$ can be found in the set
$\set{(\lambda_1^n,\ldots,\lambda_d^n)\st n\in\nat}$.  We will recall a
classical result that says that the set above is a dense subset of the
set of points in the $d$-dimensional torus that have all the
multiplicative relations as $\lambda_1,\ldots,\lambda_d$. 

Consider the set of multiplicative relations of $\lambda=(\lambda_1,\ldots,\lambda_d)$:
\begin{align*}
  \mathcal{M}_\lambda\defeq\set{\mathbf{v}\in\intg^d\st \lambda_1^{v_1}\lambda_2^{v_2}\cdots \lambda_d^{v_d}=1}. 
\end{align*}
The one-dimensional torus is the unit circle $\torus\defeq\set{z\in\com\st |z|=1}$. Define the set of points in $\torus^d$ having all the multiplicative relations of $\lambda$ as follows:
\begin{align*}
  \torus_\lambda\defeq\set{\mathbf{z}\in\torus^d\st z_1^{v_1}z_2^{v_2}\cdots z_d^{v_d}=1\text{ for all }\mathbf{v}\in\mathcal{M}_\lambda}.
\end{align*}
Denote by $s\st \torus_\lambda\to\torus_\lambda$ the map
\begin{align*}
  (z_1,\ldots,z_d)\mapsto (z_1\lambda_1,\ldots,z_d\lambda_d). 
\end{align*}
With this new notation we are interested in the set
$\set{s^n(1,\ldots,1)\st n\in\nat}$. To prove that it is a dense
subset of $\torus_\lambda$ we will use Kronecker's theorem on
simultaneous Diophantine approximation.

\begin{theorem}[{Kronecker, see \eg\cite[Page 53]{cassels1957introduction}}]
  \label{th:kronecker}
  Let $\theta_1,\ldots,\theta_k$, $\varphi_1,\ldots,\varphi_k\in \rel$ such that for any integers $a_1,\ldots,a_k$,
  \begin{align*}
    \sum_{i=1}^ka_i\theta_i \in \intg \qquad \Rightarrow \qquad \sum_{i=1}^ka_i\varphi_i\in\intg.
  \end{align*}
  Then for every $\epsilon>0$, there exists $n\in\nat$ and integers $r_1,\ldots,r_k$ such that
  \begin{align*}
    \left | n\theta_i - r_i - \varphi_i\right |\le \epsilon,
  \end{align*}
  for all $i\in\set{1,\ldots,k}$.
\end{theorem}
\begin{lemma}
  \label{lem:dense}
  For all $\mathbf{z}\in\torus_\lambda$, the set $O(\mathbf{z})\defeq\set{s^n(z_1,\ldots,z_d)\st n\in\nat}$ is dense in~$\torus_\lambda$. 
\end{lemma}
\begin{proof}
  Let $\mathbf{y}\in\torus_\lambda$. We have to prove that $O(\mathbf{z})$ intersects every $\epsilon$-ball around $\mathbf{y}$. Let us write
  \begin{align*}
    \lambda_i=e^{\theta_i\ 2\pi\ii},\qquad z_i=e^{\alpha_i\ 2\pi\ii},\qquad y_i=e^{\beta_i\ 2\pi\ii},
  \end{align*}
  and set $\varphi_i=\beta_i-\alpha_i$, for $i\in\set{1,\ldots,d}$. Because $\mathbf{y}$ and $\mathbf{z}$ belong to $\torus_\lambda$, and the multiplicative relations of $\lambda$ correspond to additive relations of $\theta$, the hypothesis of Theorem~\ref{th:kronecker} is fulfilled and the theorem can be applied. It tells us that there exists $n\in\nat$ and integers $r_1,\ldots,r_d$ such that for every $i\in\set{1,\ldots,d}$,
  \begin{align*}
    \left |z_i\lambda_i^n-y_i\right | = \left | e^{(\alpha_i+n\theta_i-r_i)\ 2\pi\ii}-e^{\beta_i\ 2\pi\ii}\right |\le 2\pi\ \left | \alpha_i+n\theta_i-r_i-\beta_i \right |\le 2\pi\epsilon. 
  \end{align*}
\end{proof}

Compactness of $\torus_\lambda$ together with Lemma~\ref{lem:dense}
entail that any open set in $\torus_\lambda$ can be reached in a
bounded number of steps from any other point.
\begin{lemma}
  \label{lem:bounded steps}
  Let $U\subseteq \torus_\lambda$ be an open set. There exists $B\in\nat$ such that for every $\mathbf{x}\in\torus_\lambda$, there exists $n\le B$ such that $s^n(\mathbf{x})\in U$. 
\end{lemma}
\begin{proof}
  Lemma~\ref{lem:dense} implies that for any $\mathbf{z}\in \torus_\lambda$ , there exists some $n\in\nat$ such that $s^n(\mathbf{z})\in U$. Whence by continuity of the successor function $s$, we have that
  \begin{align*}
    \set{s^{-n}(U)\st n\in\nat} = \torus_\lambda,
  \end{align*}
  is an open cover of $\torus_\lambda$. Since $\torus_\lambda$ is bounded and closed as a subset of $\torus^d$, it is compact. It follows that it admits a finite sub-cover, \ie there exists $B\in\nat$ such that
  \begin{align*}
    \set{s^{-n}(U)\st n\in\set{1,2,\ldots, B}}=\torus_\lambda.
  \end{align*}
\end{proof}

\subsection{The proof of \Cref{thm:simple almost periodic}}
\label{subsec:proof of simple almost periodic}

It is tempting to try to prove \Cref{thm:simple almost periodic} by
showing that the sign descriptions of every subsequence
$\lu^{(\ell)}$, where $0\le \ell < P$, is almost periodic and
combining the results. Unfortunately the proof cannot be modular in
this respect, for the simple fact that the product of two almost
periodic sequences need not be almost periodic itself; see \cite[Theorem 22]{muchnik03_almos_period_sequen}. We must directly prove almost periodicity for the whole sequence, which is done as follows. 

Let $\lu$ be a simple \lrs, $\sigma\in\{-,0,+\}^\omega$ its sign description, and $w\in\set{-,0,+}^*$, a pattern that occurs infinitely many times in $\sigma$. We have to prove that the distances between consecutive occurrences are bounded.

Since $w$ occurs infinitely many times in $\sigma$, there is some $m\in\nat$ such that for infinitely many $n$,
\begin{align}
  \label{eq:occur}
  w\text{ occurs in }\sigma_{nP}\sigma_{nP+1}\cdots \sigma_{(n+m)P-1},
\end{align}
where we recall that $P$ was defined as the least common multiple of
orders of roots of unity among the ratios of roots of $\lu$. Since
the right-hand side of~\eqref{eq:occur} is a word over a finite alphabet,
there exists a word 
$w'\in\set{-,0,+}^*$ that has $w$ as an infix such that for infinitely many $n$,
\begin{align*}
  w' = \sigma_{nP}\sigma_{nP+1}\cdots \sigma_{(n+m)P-1}.
\end{align*}
We prove that there is an upper bound for the distances among
successive such $n$, which clearly implies almost periodicity of
$\sigma$. 

Cut the word $w'$ into $m$ factors of length $P$ such that
\begin{align*}
  w'=w'(1)w'(2)\cdots w'(m). 
\end{align*}
Applying Lemma~\ref{lem:same sign} to each subsequence $\lu^{(\ell)}$,
and combining the resulting linear functions together, we obtain a linear
function $f\st \torus_\lambda\to\rel^P$ such that for all but finitely
many $n$, if $f\left(s^n(1,\ldots,1)\right)=(a_1,\ldots, a_P)$, then
\begin{align*}
  \sgn(a_1)\ \sgn(a_2)\cdots\sgn(a_P)=\sigma_{nP}\sigma_{nP+1}\cdots\sigma_{(n+1)P-1}. 
\end{align*}

While $f$ clearly maps $s^n(1,\ldots,1)$ to $\rel^P$, the reason why the same is true for other elements of $\torus_\lambda$ is as follows. The linear map in Lemma~\ref{lem:same sign} has reals as a co-domain because certain pairs among $\lambda_1,\ldots,\lambda_d$ are complex conjugates of one another, which allows for cancelling out of their imaginary parts. In every tuple in $\torus_\lambda$ the same pairs of numbers are complex conjugates of one another, since being a complex conjugate for elements of the unit circle is a multiplicative relation, and elements of $\torus_\lambda$, by definition, have all the multiplicative relations of $(\lambda_1,\ldots, \lambda_d)$. 

Denote by $g\st \torus_\lambda\to \set{-,0,+}^P$ the composition of $f$ and $\sgn$ applied component-wise.

Since $\lu^{(\ell)}$ are all non-degenerate, because of \Cref{prop:shapiro}, some coordinates of $f(s^n(1,\ldots,1))$ are identically zero, and some are zero only for finitely many $n$. Denote by $Z\subseteq\set{1,2,\ldots,P}$ the former. On components in $Z$, $f$ is a constant function mapping to zero. 

It follows from the continuity of $f$ that for elements of $v\in\set{-,0,+}^P$ that have zeros exactly in coordinates $Z$, $f^{-1}(v)$ is an open subset of $\torus_\lambda$. Since $w'(1),\ldots,w'(m)$ are words that occur infinitely often, they must have zeros exactly in positions in $Z$, hence the set
\begin{align*}
  U(w')\defeq\set{\mathbf{x}\in\torus_\lambda\st g(\mathbf{x})=w'(1), g\left(s(\mathbf{x})\right)=w'(2),\ldots, g\left(s^{m-1}(\mathbf{x})\right)=w'(m)},
\end{align*}
is open. By applying Lemma~\ref{lem:bounded steps} we know that there exists $B\in\nat$ such that for any $\mathbf{y}\in\set{s^n(1,\ldots,1)\st n\in\nat}$, there exists $k\le B$ with $s^k(\mathbf{y})\in U(w')$. So from any point, in fewer than $B$ steps, we enter the set $U(w')$ from where $g$ outputs $w'$ (in the next $m$ steps). This proves that the distances between consecutive $n$ for which \eqref{eq:occur} holds is at most $B\cdot P$. 
\subsection{Effectiveness}

We make a closer inspection of the proof of almost periodicity above, in order to gather three lemmas which pull out what can be effectively computed about the sign description.

\begin{lemma}
  \label{lem:nonemptyU}
  Let $w'=w'(1)w'(2)\cdots w'(m)$  be such that $w'(i)$ are factors of length $P$, then the following two statements are equivalent
  \begin{itemize}
  \item for infinitely many $n$,
  \begin{align*}
    \sigma_{nP}\sigma_{nP+1}\cdots \sigma_{(n+m)P-1}=w'
  \end{align*}
  \item $U(w')$ is non-empty.
  \end{itemize}
\end{lemma}
\begin{proof}

  \noindent($\Rightarrow$) Let $k\in\nat$ be such that the equation in \Cref{lem:same sign} holds for all $\ell\in\set{0,1,\ldots,P-1}$. From the hypothesis there exists some $n>k$ such that $\sigma_{nP}\cdots \sigma_{(n+m)P-1}=w'$. Now the definition of the set $U(w')$ implies that it is not empty.

  \noindent($\Leftarrow$) Since $U(w')$ is non-empty and open, we can apply \Cref{lem:bounded steps}, which gives us a bound $B$ on how many steps we have to take in the walk in $\torus^d$ before we enter again the set $U(w')$. Therefore we enter the set $U(w')$ infinitely many times, and hence the word $w'$ occurs infinitely often in~$\sigma$. 
\end{proof}

The next lemma says that modulo a finite prefix, the word $\sigma$ is strongly recurrent, which means that if some word occurs in it, it does so infinitely often. This stems from the fact that after some threshold, the sign description only depends on the walk in $\torus^d$, which is repetitive.

\begin{lemma}
  \label{lem:occurinf}
  There exists a threshold $c\in\nat$ such that any word that occurs in the suffix $\sigma_c\sigma_{c+1}\cdots$, occurs infinitely often in $\sigma$. 
\end{lemma}
\begin{proof}
  Let $n_1\in\nat$ be such that for all $n\ge n_1$ and $0\le \ell < P$, the equation in \Cref{lem:same sign} holds. Let $n_2\in\nat$ be such that for all $0\le\ell < P$, $\lu^{(\ell)}$ is either identically zero or has no zeros after $n_2$ (well defined thanks to \Cref{prop:shapiro}). Set $c=\max\set{n_1,n_2}$. Let $w$ be some word that occurs after the threshold $c$ in $\sigma$. Then there is some $n$ and $m$ such that $w$ occurs in
\begin{align*}
  \sigma_{nP}\sigma_{nP+1}\cdots \sigma_{(n+m)P-1}.
\end{align*}
Call this word $w'$ and let $w'(1),\ldots,w'(m)$ be its decomposition into factors of length $P$. Since $w'$ occurs after $n_2$ the factors $w'(1), \ldots, w'(m)$ have zeros exactly in the same positions $Z\subseteq\set{1,\ldots , P}$. This, together with the fact that $w'$ occurs after $n_1$ implies that $U(w')$ is open and non-empty. Now \Cref{lem:bounded steps} gives us a bound for the distances between consecutive occurrences of $w'$. 
\end{proof}

The last crucial property is that for all $w$ the set $U(w)$ is semi-algebraic and effective. We give here a sketch and relegate the full proof, as well as the relevant definitions, to \Cref{sec:foth}.

\begin{restatable}{lemma}{fodeflem}
  \label{lem:fodef}
  For all $w$, $U(w)$ is semi-algebraic, and we can compute the first-order formula that defines it. 
\end{restatable}
\begin{proof}[Proof sketch]
  The set $U(w)$ is a subset of $\com^d$ while semi-algebraic sets are subsets of $\rel^n$. However there is a simple first-order interpretation of $\com$ in $\rel$, we take for every complex number two real variables, one for the real part and one for the imaginary part.

  The set of normalized roots $\lambda_1, \ldots, \lambda_d$ are algebraic numbers whose first-order formulas we can effectively construct given a \lrs. Their multiplicative relations, \ie the set $\mathcal{M}_\lambda$ has a finite basis, which can be computed using a result of Masser. Whence it follows that $\torus_\lambda$ is semi-algebraic and that we can effectively construct the first-order formula that defines it.

  The lemma now follows because the coefficients of the linear map $f\st \torus_\lambda\to\rel^P$ (in the definition of $U(w)$ and $g$), are algebraic and we can compute their first-order formula. For a full proof see~\Cref{sec:foth}.
\end{proof}

\section{A counter-example for general \lrs}
\label{sec:cexample}
We have proved that simple \lrs have almost periodic sign
descriptions. In this section we show that the same does not hold for
general \lrs.  The additional structure of simple \lrs is consonant
with what is known about decidability: e.g., it is decidable whether a
simple \lrs is ultimately positive, whereas the decidability of the
same question for general \lrs is open and a positive result would
imply computability of Lagrange constants of certain transcendental
numbers, see~\cite[Theorem
  5.1]{ouaknine13_posit_probl_low_order_linear_recur_sequen}.

Let $\lambda\in\torus$ be any algebraic number in the unit circle that is not a root of unity. Consider the generalized power sum
\begin{align*}
  u_n \defeq \frac n 2\ \lambda^n + \frac n 2\ \overline{\lambda}^n + (1-n)\ 1^n,
\end{align*}
where $\overline{\lambda}$ is the complex conjugate of $\lambda$. As discussed before, such sums are equivalent to linear recurrence sequences (see \cite[Section 1.1.6]{recseq}), so we can easily extract an \lrs of order six from the sum above.
This example has been designed in such a way as to have the following property. Set $\theta:=\arg(\lambda)$, and using Euler's formula deduce that $u_n=1-n+n\cos(n\theta)$. Consequently for all $n$, we have: 
\begin{align}
  \label{eq:cexample}
  u_n>0\qquad \Leftrightarrow\qquad \cos(n\theta)>1-\frac 1 n.
\end{align}
We will prove that this sequence has a sign description that is not almost periodic. More precisely we will prove that (a) the letter `$+$' occurs infinitely often in the sign description and (b) that the distances between consecutive occurrences can be arbitrarily large.

The intuition is as follows. Since $\lambda$ is not a root of unity,
$\theta$ is not a rational multiple of $\pi$ and hence
$\set{\cos(n\theta)\st n\in\nat}$ is a dense subset of $[0,1]$. Using
basic properties of the cosine function we can prove that the right-hand inequality
in~\eqref{eq:cexample}, \ie $\cos(n\theta)>1-1/n$, is true for
infinitely many $n$. However, since the interval~$(1-1/n,1]$ becomes
  arbitrarily tight as $n$ increases, we have to wait longer and longer
  until $\cos(n\theta)$ enters it.

We give now the proofs of the two claims (a) and (b) above.

\begin{proposition}
  For infinitely many $n\in\nat$, $u_n>0$.
\end{proposition}
\begin{proof}
  For a real number $x\in\rel$ denote by $[x]$ its distance to the closest integer, and by $[x]_{2\pi}$ its distance to the closest integer multiple of $2\pi$, \ie
  \begin{align*}
    [x]\defeq \min_{k\in\intg}|x-k|,\qquad\qquad [x]_{2\pi}\defeq \min_{k\in\intg}|x-2k\pi|.
  \end{align*}
  We will first prove that for infinitely many $n\in\nat$,
  \begin{align*}
    [n\theta]_{2\pi}<\frac {2\pi} n.
  \end{align*}
  This is a corollarly of Dirichlet's Theorem~\cite[Chapter 2, Theorem
    1]{lang95_introd_dioph_approx}, which states that for every
  $x\in\rel$ there exist infinitely many $n\in\nat$ such that
  $[nx]<1/n$.  Indeed since for any $x\in\rel$,
  $[x]_{2\pi}=2\pi[x/(2\pi)]$, Dirichlet's Theorem implies that for
  infinitely many $n\in\nat$,
  \begin{align}
    \label{eq:cexample bound}
    [n\theta]_{2\pi}=2\pi\left[n\frac \theta {2\pi}\right]<\frac {2\pi} n.
  \end{align}
By the monotonicity of the cosine function on $[0,\pi]$ we have that
for all $n \geq 2$, $[n\theta]_{2\pi}<2\pi/n$ if and only if
$\cos(n\theta)>\cos(2\pi/n)$. As a consequence of \eqref{eq:cexample
  bound}, the inequality $\cos(n\theta)>\cos(2\pi/n)$ holds for
infinitely many $n$.

  Using a Taylor series expansion of cosine, we can prove that for $x$ sufficiently close to $0$,
  \begin{align*}
    \cos x\ge 1-\frac {x^2} 2\ge 1-\frac {|x|}{2\pi}. 
  \end{align*}
  Applying this bound to $x=2\pi/n$, we derive that for infinitely many $n\in\nat$,
\begin{align*}
  \cos(n\theta)>\cos\left(\frac {2\pi} n\right) \ge 1-\frac 1 n. 
\end{align*}
Therefore, from \eqref{eq:cexample} it follows that $\sq u$ is positive in infinitely many positions, or equivalently the letter `$+$' occurs infinitely often in its sign description.
\end{proof}

\begin{proposition}
  For every $p\in\nat$ we can find $p$ consecutive entries in $\sq u$ that are negative or zero. 
\end{proposition}
\begin{proof}
  Fix some $p\in\nat$ and let $N\in\nat$ be such that if $\cos x>1-1/N$ then $x\le 2\pi/(p+1)$, for all $x\in[-\pi,\pi]$. Define:
  \begin{align*}
    I_0\defeq\left\{z\in\torus\st \cos(\arg z)>1-\frac 1 N,\ -\pi\le \arg z \le\pi\right\}.
  \end{align*}
  Using \eqref{eq:cexample}, clearly for all $n\ge N$,
  \begin{align*}
    u_n>0\qquad\Rightarrow\qquad e^{n\theta\ii}=\lambda^n\in I_0.
  \end{align*}
  For all $k\in\nat$, we denote by $I_k$ the rotation of $I_0$ by $-k\theta$, \ie $I_k\defeq e^{-nk\theta\ii}I_0$. Similarly to above, for any $k\in\nat$ and $n\ge N$,
  \begin{align*}
    u_{n+k}>0\qquad \Rightarrow\qquad e^{n\theta\ii}=\lambda^n\in I_k.
  \end{align*}
  Define $I\defeq I_0\cup I_1\cup \cdots \cup I_{p-1}$, so that for all $n\ge N$,
  \begin{align}
    \label{eq:contrapos}
    u_{n+k}>0\text{ for some $k\in\set{0,\ldots,p-1}$}\qquad \Rightarrow\qquad e^{n\theta\ii}=\lambda^n\in I.
  \end{align}
  By construction of $N$ above, the set $\torus\setminus I$ is non-empty and it is a finite union of intervals, so in particular it has non-empty interior $U$. It follows from an analogue of \Cref{lem:dense} that for some $n\in\nat$, $\lambda^n\in U$. This in turn means, using the contrapositive of \eqref{eq:contrapos}, that starting from $u_n$, the next $p$ consecutive entries are either zero or negative. 
\end{proof}


\section{Deciding $\omega$-regular properties}
\label{sec:omegareg}
In this section we prove the main theorem, which we recall here.
\mainth*

We will first explain what prefix-independent languages are, and how does an automaton that accepts such a language look like. It will be simpler to work with a finite monoid (similar to the syntactic monoid) that has all the relevant information about the automaton, so we will define this afterwards. In the end of the section we will describe the algorithm. 

Prefix-independent languages are those that have the property that by modifying a word in finitely many places it is not possible to change its membership in the language. More precisely:

\begin{definition}[Prefix-independent language]
  A language $\mathcal L\subseteq \Sigma^\omega$ is prefix-independent if for all infinite words $\alpha,\alpha'$ such that we can get $\alpha'$ from $\alpha$ with finitely many insertions and deletions we have that
  \begin{align*}
    \alpha\in\mathcal L\qquad \Leftrightarrow\qquad \alpha'\in\mathcal L.
  \end{align*}
\end{definition}
These are languages that have trivial right-congruence \cite[Section 4]{angluin20_regul_om_with_infor_right_congr}. 

We will assume that the $\omega$-regular language $\mathcal L$ is given as a deterministic M\"uller automaton~$\aut$, which is a tuple:
\begin{align*}
  \langle\underbrace{Q}_{\text{states}}, \qquad \underbrace{q_0}_{\text{initial state}}, \qquad \underbrace{\Sigma}_{\text{alphabet}}, \qquad \underbrace{\delta\st Q\times \Sigma\to Q}_{\text{transition function}}, \qquad \underbrace{F\subseteq\pow(Q)}_{\text{accepting states}}\ \rangle.
\end{align*}
The automaton accepts a word $\alpha$ if the unique run is such that the set of states that appear infinitely often belongs to $F$. We also assume that any state can be reached from the initial state. 
\begin{lemma}
  \label{lem:allany}
  Let $\aut$ be a deterministic M\"uller automaton that recognizes a prefix-independent language, and $\alpha$ an infinite word. Then the following are equivalent:
  \begin{enumerate}
  \item $\aut$ accepts $\alpha$, 
  \item $\aut$ accepts some suffix of $\alpha$ starting from some state,
  \item $\aut$ accepts every suffix of $\alpha$ starting from any state. 
  \end{enumerate}
\end{lemma}
\begin{proof}
  We prove $2\Rightarrow 3$, other directions are trivial. Let $\beta$ be a suffix of $\alpha$ that is accepted starting from some state. Since the latter can be reached from $q_0$, it follows that $w\alpha$ is accepted by the automaton (started in $q_0$) for some finite word $w$. For 3, take any suffix $\beta'$ of $\alpha$ and any state $q_1$; and let $w'$ be the word that takes the automaton from state $q_0$ to state $q_1$. The words $w\beta$ and $w'\beta'$ have a suffix in common, since the language is prefix-independent, $w'\beta'$ is accepted by the automaton. 
\end{proof}

As a consequence of this lemma, to show that the automaton accepts the sign description, we only have to prove that \emph{some} suffix of the sign description $\sigma$ is accepted by the given automaton $\aut$ started at \emph{some} state $q_1$.

Before we turn our attention to the monoid associated with the automaton, we gather here from the prequel some properties of the sign description of a given simple \lrs. These are the essential properties that will be used by the algorithm.

\begin{proposition}
  \label{prop:properties}
  Let $\lu$ be a \lrs and $\sigma$ its sign description, then the following hold:
  \begin{enumerate}
  \item $\sigma$ is almost-periodic,
  \item there is a threshold $c\in\nat$ such that any word that occurs in the suffix $\sigma_c\sigma_{c+1}\cdots$, occurs infinitely often in $\sigma$,
  \item there is a procedure that inputs a finite word $w$ and decides whether $w$ occurs infinitely often in~$\sigma$,
  \item there is a procedure that inputs a finite word $w$ that occurs infinitely often in $\sigma$ and outputs the bound on the distances between consecutive occurrences. 
  \end{enumerate}
\end{proposition}
\begin{proof}
    Property 1 is \Cref{thm:simple almost periodic}, Property 2 is \Cref{lem:occurinf}. We prove 3.

    Recall that $P$ is the least common multiple of orders of roots of unity among the ratios of roots of the given \lrs. We can effectively determine it, and moreover for $\ell\in\set{0,1,\ldots,P-1}$ we can decide which subsequence $\lu^{(\ell)}$ is identically zero; which we denote by $Z\subseteq\set{0,1,\ldots, P-1}$. We construct a word $w'$ such that 
    \begin{align*}
      w'=w'(1)w'(2)\cdots w'(m),\text{ for some $m\in\nat$, }
    \end{align*}
    where the factors $w'(i)$ are of length $P$, $w$ occurs in $w'$, and the factors $w'(i)$ have zeros exactly in positions $Z$. If this is not possible, the procedure returns \textbf{no}. The word $w$ occurs infinitely often if and only if for infinitely many $n$, we have
    \begin{align*}
      \sigma_{nP}\sigma_{nP+1}\cdots\sigma_{(n+m)P-1}=w'. 
    \end{align*}
    The latter is true if and only if $U(w')$ is non-empty, according to \Cref{lem:nonemptyU}. Since $U(w')$ is semi-algebraic (~\Cref{lem:fodef}) and we can effectively construct its formula, to test whether $U(w')$ is empty we can use Tarski's algorithm, see~\Cref{th:tarski}.

    We now prove item 4. We continue as above and define $w'$, so that $U(w')$ is non-empty. Since the normalized roots $\lambda_1,\ldots, \lambda_d$ are algebraic with effective descriptions, for all $k\in\nat$, the set of points $\mathbf z\in\torus_\lambda$ such that $(z_1\lambda_1^k,\ldots,z_d\lambda_d^k)\in U(w')$ is semi-algebraic, and we can effectively construct its formula $\varphi_k$. The formula for the set of points that enter $U(w')$ in at most $j$ steps, $\varphi_{\le j}$, is just the disjunction of formulas $\varphi_1,\ldots,\varphi_j$. Since $U(w')$ is open and non-empty, \Cref{lem:bounded steps} implies that there exists some $B$ such that any point in $\torus_\lambda$ enters $U(w')$ in fewer than $B$ steps. In the language above this means that $\Phi(B)$ which says:
    \begin{align*}
      \text{every element of $\torus_\lambda$ satisfies $\varphi_{\le B}$}, 
    \end{align*}
    is true. Now to compute $B$, or a different (stronger) bound, we only need to find the first formula in $\langle\Phi(1),\Phi(2),\ldots\rangle$ that is true. We can do this using Tarski's algorithm.
  \end{proof}

  Fix a deterministic M\"uller automaton $\aut$ and a \lrs $\lu$ with sign description $\sigma$ for the rest of this section. We provide a ``wrapper'' for the procedures in properties 3 and 4 in the proposition above.  There is a procedure ``$\inter$'' that inputs a word $w$ and outputs a finite set of words, or \textbf{no}:
  \begin{align*}
    \red w\mapsto\begin{cases}
      \text{\textbf{no}}\qquad &\text{if $w$ does not occur infinitely often in $\sigma$,}\\
      \set{w_1,w_2,\ldots,w_k}\qquad &\text{otherwise},
      \end{cases}
  \end{align*}
  such that, in the case when $w$ occurs infinitely often in $\sigma$, 
  \begin{align*}
    \sigma=r\ \red w\ w_{i_1}\ \red w\ w_{i_2}\ \cdots,
  \end{align*}
  where $r$ is some finite prefix and $i_1,i_2,\ldots$ take values in $\set{1,2,\ldots, k}$.

Intuitively, from the almost periodicity of $\sigma$, when $w$ occurs infinitely often, the distance between the occurrences is bounded, hence there can be only finitely many words that appear between consecutive occurrences of $w$, using the procedure in Property 3 of \Cref{prop:properties}, we can find these words that appear between occurrences of $w$, and it is this set of words that the procedure ``$\inter$'' returns. 

\subsection{The finite monoid associated to $\aut$}

Denote by $\mon$ the following monoid. Its elements are directed and labeled graphs, where the set of vertices is $Q$ (the set of states of the automaton), and the edges are labeled by subsets of $Q$. The product of the element $\mx$ with the element $\my$ is defined as follows: for some $q_1,q_2\in Q$ and $S_1,S_2\subseteq Q$
\begin{align*}
  \underbrace{q_1\xrightarrow{S_1\cup S_2}q_2}_{\text{in $\mx\cdot\my$}},
\end{align*}
if and only if there exists some $q'\in Q$ such that:
\begin{align*}
  \underbrace{q_1\xrightarrow{S_1}q'}_{\text{in $\mx$}}\, \text{ and }\, \underbrace{q'\xrightarrow{S_2}q_2}_{\text{in $\my$}}
\end{align*}

The homomorphism $h$, is defined as follows. For any letter $a$ of the alphabet (which in our case is $\set{-,0,+}$), $h(a)$ is such that for all $q_1,q_2\in Q$
\begin{align*}
  \underbrace{q_1\xrightarrow{\set{q_1,q_2}} q_2}_{\text{ in $h(a)$}},
\end{align*}
if and only if there is a transition in the automaton $\aut$ from $q_1$ to $q_2$ with the letter~$a$. The monoid $\mon$ is the monoid that is generated by $\set{h(a)\st a\in \set{-,0,+}}$ as well as $h(\epsilon)$, where $\epsilon$ is the empty word, for the neutral element.

This monoid (a variant of the transition semigroup \cite[Chapter 3]{04_infin_words_autom_semig_logic_games}) gathers all the information needed from the automaton $\aut$, \eg if there is an edge from $s_1$ to $s_2$ in $h(w)$, labeled by $S$, it means that in the automaton, we can go from state $s_1$ to state $s_2$ with the word $w$ while visiting all the states in $S$.

In our case, where the automaton $\aut$ is deterministic, the elements of the monoid $\mon$ are particularly simple in that in every $\mx$, and every $q\in Q$, there is only one outgoing edge from $q$. Therefore it makes sense to talk about the {\em states that are seen from $q$} in $\mx$, \ie the set $S$ that is the label of the unique outgoing transition of from $q$ in $\mx$.

\begin{definition}[Increasing product] Let $\mx_1,\ldots,\mx_k\in\mon$. We say that the product 
\begin{align*}
  \mx=\mx_1\mx_2\cdots \mx_k
\end{align*}
is increasing, if there exists some $q\in Q$ such that the states that are seen from $q$ in $\mx_1$ is a strict subset of the states that are seen from $q$ in $\mx$. 
\end{definition}

In terms of the automaton $\aut$, this definition means that we visit strictly more states by reading a word associated to $\mx_2\cdots\mx_k$, than we do by reading a word associated to $\mx_1$.

\begin{example}
  Consider the following elements. 
  \begin{center}
    \includegraphics[width=0.8\textwidth]{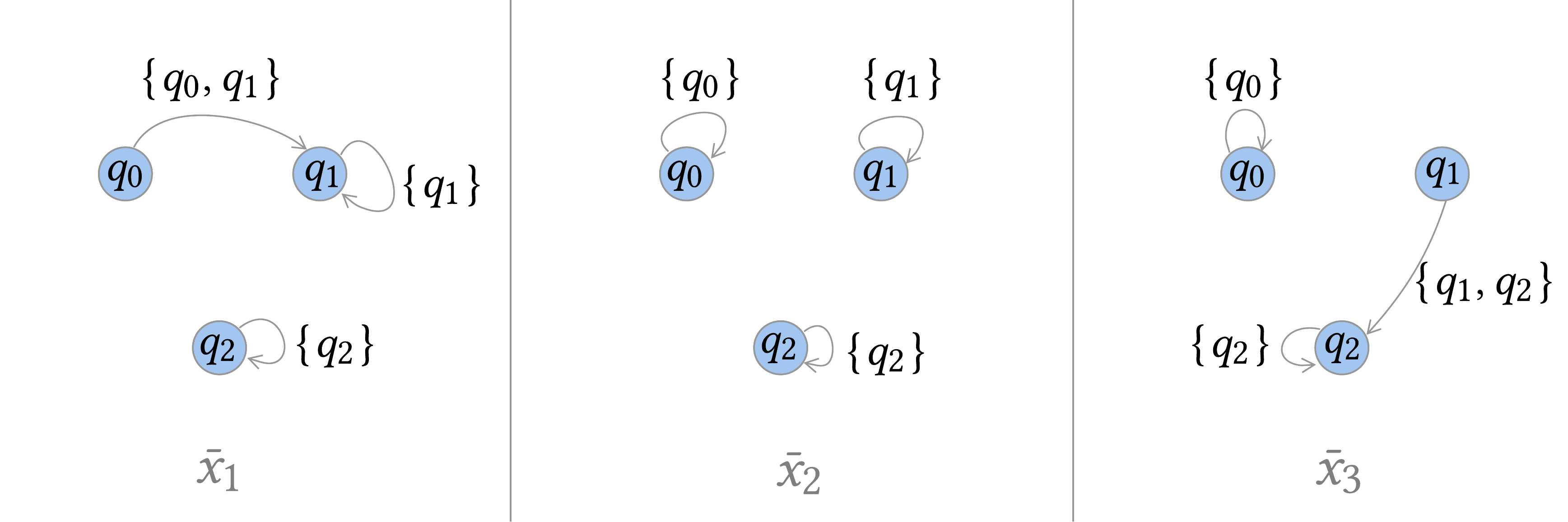}
  \end{center}
  The product $\mx_1\mx_2$  is not increasing, however the product $\mx_1\mx_2\mx_3$ is increasing, because of the path:
  \begin{align*}
    q_0\xrightarrow{\set{q_0,q_1}}q_1\xrightarrow{\set{q_1}}q_1\xrightarrow{\set{q_1,q_2}}q_2.
  \end{align*}
  So in the product $\mx_1\mx_2\mx_3$, the states that are seen from $q_0$ are $\set{q_0,q_1,q_2}$. But in $\mx_1$, the states that are seen from $q_0$ are $\set{q_0,q_1}$, a strict subset. Similarly, the product $\mx_2\mx_3$ is also increasing. 
\end{example}

We make the following observation about increasing products before we move on to the description of the algorithm.

\begin{lemma}
  \label{lem:shortinc}
  Let $\mx_1\mx_2\cdots \mx_k$ be an increasing product. Then there exists $i\in\set{1,\ldots,k-1}$ such that:
  \begin{align*}
    \mx_{i}\mx_{i+1}\text{ is increasing.}
  \end{align*}
  Moreover, for all $1\le r \le i$ and $i< r'\le k$ the product
  \begin{align*}
    \mx_{r}\mx_{r+1}\cdots\mx_{r'}\text{ is increasing.}
  \end{align*}
\end{lemma}
\begin{proof}
  Since $\mx_1\mx_2\cdots \mx_k$ is increasing there exists a state $q$ and a path 
  \begin{align*}
    \underbrace{q\xrightarrow{S_1}q_1}_{\text{in $\mx_1$}}\, \underbrace{\xrightarrow{S_2}q_2}_{\text{in $\mx_2$}}\cdots\,\, \underbrace{\xrightarrow{S_j}q_j}_{\text{in $\mx_j$}},
  \end{align*}
  such that $S_j$ is the first set that is not a subset of $S_1$. Let $i\defeq j-1$. Now the lemma follows from the fact that multiplying an increasing product to the right with elements of the monoid gives us again an increasing product. 
\end{proof}
\subsection{Description of the algorithm}
The algorithm only uses the monoid $\mon$ and the procedure $\inter$. It starts with some word $w$ that occurs infinitely often in $\sigma$ (\eg a letter). Let $\inter(w)=\set{w_1,\ldots,w_k}$. Among words:
\begin{align*}
  \red w\ w_i\ \red w\ w_j,
\end{align*}
where $i,j\in\set{1,\ldots,k}$ that occur infinitely often in $\sigma$, it tries to find one such that
\begin{align*}
  h(\red w)\ h(w_i)\ h(\red w)\ h(w_j)\text{ is increasing.}
\end{align*}
If it manages to find such a word $ww_iww_j$, it calls $\inter(ww_iww_j)$ and repeats the steps above. Since for every state $q\in Q$, the states seen from $q$ in $h(w)$ is a subset of the states that are seen from $q$ in $h(ww_iww_j)$, it follows that in fewer than $|Q|^2$ iterations, a fix-point is reached: a word $u$, with $\inter(u)=\set{u_1,\ldots,u_c}$ such that among all words
\begin{align*}
  \red u\ u_i\ \red u\ u_j,
\end{align*}
that occur infinitely often in $\sigma$,
\begin{align}
  \label{eq:notinc}
  h(\red u)\ h(u_i)\ h(\red u)\ h(u_j)\text{ is {\em not }increasing.}
\end{align}
We give the crucial property of the word $u$ before proceeding with the last step of the algorithm.

Since $\inter(u)=\set{u_1,\ldots,u_c}$ there exists some suffix of $\sigma$ that is equal to
\begin{align*}
  \sigma'\defeq\red u\ u_{i_1}\ \red u\ u_{i_2}\cdots,
\end{align*}
where $i_n$ take values in $\set{1,\ldots, c}$, and moreover from Property 2 in \Cref{prop:properties}, we can assume that $\sigma'$ is such that any word that occurs in $\sigma'$, occurs infinitely often.

Observe that for all $n\in\nat$
\begin{align}
  \label{eq:notinc2}
  h(\red u)\ h(u_{i_1})\ h(\red u)\ h(u_{i_2})\cdots h(\red u)\ h(u_{i_n})\text{ is {\em not }increasing.}
\end{align}
Indeed, if it were an increasing product, we would be able to find a short one, because of \Cref{lem:shortinc}, in particular we would be able to find an increasing product of the form $h(u)h(u_i)h(u)h(u_j)$. By construction of $\sigma'$ the word $uu_iuu_j$ occurs infinitely often in $\sigma'$ (and therefore also in $\sigma$), which contradicts \eqref{eq:notinc}.

In the last step, the algorithm picks an edge
\begin{align*}
  q\xrightarrow{S}q'
\end{align*}
in $h(u)$, such that $S$ has minimal cardinality out of all the other labels of edges in~$h(u)$. It returns \textbf{yes} if and only if $S\in F$, where $F$ is the collection of accepting sets of states in the definition of the automaton $\aut$.

For the correctness of the algorithm we prove the following claim.
\begin{lemma}
  When the automaton $\aut$ reads the word $\sigma'$ starting from the state $q$, the set of states that is seen infinitely often is $S$.
\end{lemma}
This suffices because the automaton $\aut$ accepts the infinite word $\sigma$ if and only if it accepts the word $\sigma'$ starting from state $q$, thanks to \Cref{lem:allany}.

\begin{proof}
  Observation \eqref{eq:notinc2} implies that the run of the automaton $\aut$ starting from $q$ and reading the infinite word $\sigma'$ looks as follows:
  \begin{align*}
    q\ \xrightarrow{\text{visits }\red S} q_1\ \xrightarrow{\text{visits } S_1} q_2 \xrightarrow{\text{visits }\red S} q_3 \xrightarrow{\text{visits }S_2} q_4 \xrightarrow{\text{visits }\red S} \cdots,
  \end{align*}
  where $S_n\subseteq S$; and the reason why in every second edge the label is $S$ is because there is no label that is a subset of $S$ in $h(u)$ (except itself), because we have chosen $S$ to have minimal cardinality out of all other labels in $h(u)$.
\end{proof}

We have thus proved \Cref{th:main}. As for \Cref{th:secondmain}, the circumstances are simpler. When the \lrs $\lu$ is fixed, so is the bound $\tilde n_0$, defined as the maximum of the bounds from \Cref{lem:same sign} over all $0\le\ell<P$. The number $\tilde n_0$ together with Property 4 in \Cref{prop:properties} mean that for any word $w$ we can effectively determine the bound $p$ in~\Cref{def:ap}. As a consequence, we can apply~\Cref{thm:semenov}. 

\section{Extensions}
\label{sec:extensions}
\Cref{th:main} can be extended in a couple of ways: (a) it is possible to decide properties of the product of multiple sign descriptions, corresponding to different \lrs, and (b) instead of predicates that speak about the sign of the entries, we can have general polynomial inequalities. We have chosen to present the specialized~\Cref{th:main}, so as not to obscure the principal ideas.

\subsection{Multiple \lrs}
We are given $m$ simple linear recurrence sequences:
\begin{align}
  \label{eq:all lu}
  \lu^{(1)},\lu^{(2)},\ldots,\lu^{(m)},
\end{align}
of orders $d^{(1)}, d^{(2)},\ldots, d^{(m)}$ respectively. The product of their sign descriptions $\sigma^{(1)}\times \cdots \times \sigma^{(m)}$ is an infinite word $\tilde \sigma$ over the alphabet~$\set{-,0,+}^m$. We explain how the proof in the prequel can be adapted to prove a generalisation of~\Cref{th:main},  where $\lu$ is replaced by~\eqref{eq:all lu}, and instead of the sign description $\sigma$, we have $\tilde \sigma$, the product of sign descriptions.

Define:
\begin{align*}
  \tilde\lambda\defeq\left(\lambda_1^{(1)},\ldots,\lambda_{d^{(1)}}^{(1)},\lambda_1^{(2)},\ldots,\lambda_{d^{(2)}}^{(2)},\ \ \ldots\ \ , \lambda_1^{(m)},\ldots,\lambda_{d^{(m)}}^{(m)}\right), 
\end{align*}
where $\lambda_i^{(k)}$ are the normalized roots of $\lu^{(k)}$ as in~\eqref{eq:normalized roots}. The subgroup of the torus, $\torus_{\tilde \lambda}$, and the successor function $\tilde s\st \torus_{\tilde\lambda}\to\torus_{\tilde\lambda}$ are defined as expected\footnote{There is another option of defining $\torus_{\tilde\lambda}$ as the product of $\torus_{\lambda^{(i)}}$. However this is a different object, not suitable for our needs; in particular the hypothesis of \Cref{th:kronecker} is invalidated. }. \Cref{lem:dense} can be applied to $\torus_{\tilde\lambda}$, consequently the set $\set{\tilde s^n(1,\ldots,1)\st n\in\nat}$ is a dense subset of $\torus_{\tilde\lambda}$, and any open subset of the torus can be reached in bounded number of steps.

Let $k\in\set{1,\ldots,m}$, and denote by $P^{(k)}$ the least common multiple of the orders of all roots of unity among the ratios of the roots of $\lu^{(k)}$, as defined in the beginning of~\Cref{sec:ap}. Applying~\Cref{lem:same sign} to the subsequences of~ $\lu^{(k)}$ and combining the results yields the following. There exists a linear function $f$, from $\torus_{\lambda^{(k)}}$ to $\rel^{P^{(k)}}$ and a threshold $n_0\in\nat$  such that for all $n\ge n_0$, 
\begin{align*}
  \sgn\bigg(f\big(s^n(1,\ldots,1)\big)\bigg) = \sigma^{(k)}\big[nP^{(k)},\ (n+1)P^{(k)}\big],
\end{align*}
where $\sgn$ is applied component-wise and the word on the right hand side is the factor of $\sigma^{(k)}$ that starts in position $nP^{(k)}$ and ends in position $(n+1)P^{(k)}$. Let $\tilde P$ be the least common multiple of $P^{(1)}, \ldots, P^{(m)}$. We can glue together the linear functions above to build another linear function $g\st \torus_{\tilde\lambda}\to\rel^{\tilde P}$ such that the following holds. There exists a threshold, after which, for all $n$,
\begin{align*}
  \sgn\bigg(g\big(\tilde s^n(1,\ldots,1)\big)\bigg) = \tilde\sigma\big[n\tilde P,\ (n+1)\tilde P\big].
\end{align*}

After establishing this fact, the rest of the proof depends only on the objects $\torus_{\tilde \lambda}$, $\tilde s$ and the linear function $g$. As a consequence we have the following theorem.

\begin{theorem}
  \label{th:extension1}
  Given a prefix-independent $\omega$-regular language $\mathcal L$ and simple linear recurrence sequences $\lu^{(1)}, \ldots, \lu^{(m)}$, it is decidable whether the product of sign descriptions $\tilde \sigma$, belongs to $\mathcal L$.
\end{theorem}
\subsection{Semi-algebraic predicates}
\label{subsec:sa predicates}
\Cref{th:extension1} already allows us to decide properties such as: ``only finitely many times, does the sequence pass from a value smaller than 3, to a value larger than 10, without having the value 5 in between''.

\begin{example}
More precisely, there are only finitely many $n\in\nat$ such that 
\begin{align}
  \label{eq:prop example}
  \exists m>n, u_n<3 \text{ and }u_m>10,\text{ moreover for every }r, n\le r\le m, u_r\ne 5. 
\end{align}
This can be done as follows. Define:
\begin{align*}
  u^{(1)}\defeq \sq{u_n-3},\qquad u^{(2)}\defeq \sq{u_n-5},\qquad u^{(3)}\defeq \sq{u_n-10}.
\end{align*}
These sequences are linear recurrence sequences, because \lrs are closed under addition and product of sequences. Consider the product $\tilde \sigma$, of sign descriptions of the \lrs above. This is an infinite word over the alphabet $\set{-,0,+}^3$, so a letter looks like: $(-,+,-)$. The property~\eqref{eq:prop example}, in terms of $\tilde \sigma$, can be expressed as: ``only finitely many times does a letter of the type $(-,*,*)$ appear followed by a letter of the type $(*,*,+)$ without having a letter of type $(*,0,*)$ in between''. This is a prefix-independent $\omega$-regular property.
\end{example}

In the example above we have shifted the sequence by constants $3,5$, and $10$, to produce new \lrs. More generally, \lrs are closed under sequence addition and product, as it can readily be seen when the sequences are given as exponential polynomials.
\begin{proposition}
  \label{prop:lrs closure}
  Let $\lu,\lv$ be two linear recurrence sequences. Both $\sq{u_n+v_n}$ and $\sq{u_nv_n}$ are linear recurrence sequences as well. Furthermore if both $\lu$ and $\lv$ are simple, then so is the point-wise sum and product. 
\end{proposition}
As a consequence of this proposition, given \lrs $\lu^{(1)},\ldots,\lu^{(m)}$, and a polynomial $F\in\intg[x_1,\ldots,x_m]$, the sequence $\lu(F):=\sq{F(u_n^{(1)},\ldots,u_n^{(m)})}$ is a \lrs. It follows that we can have predicates that are polynomial inequalities. In other words, for polynomials $F_1,\ldots, F_k$, we can construct the sequences $\lu(F_1),\ldots,\lu(F_k)$ and apply \Cref{th:extension1}. 

We can go one step further, and define predicates which are equal to membership in a semi-algebraic set (\Cref{sec:foth}). We explain this more precisely.
Let 
\begin{align*}
  S_1,\ldots, S_k\subseteq \rel^{m},
\end{align*}
be semi-algebraic sets. Define $\mathbf{S_i}$, $i\in\set{1,\ldots, k}$, a predicate on naturals, to be true for $n\in\nat$ if and only if
\begin{align*}
  \left(u_n^{(1)}, \ldots, u_n^{(m)}\right) \in S_i. 
\end{align*}
The \emph{zone description} of $\lu^{(1)},\ldots,\lu^{(m)}$ with respect to $\mathbf{S_1},\ldots, \mathbf{S_k}$ is the infinite word $\tau$ over the alphabet~$\pow(\set{\mathbf{S_1},\ldots,\mathbf{S_k}})$, defined for all $n\in\nat$ as:
\begin{align*}
  \text{$\tau_n$ is the subset of predicates $\set{\mathbf{S_1},\ldots,\mathbf{S_k}}$ that are true in $n$.}
\end{align*}

Since quantifiers can be eliminated in first order logic of real closed fields, membership in a semi-algebraic set reduces to fulfilling a finite set of polynomial inequalities. As we have already shown how we are able to apply \Cref{th:extension1} on predicates that are polynomial inequalities we have the following theorem.

\begin{theorem}
  \label{th:extension2}
  Given semi-algebraic sets $S_1,\ldots,S_k$, a prefix-independent $\omega$-regular language $\mathcal L$, and simple linear recurrence sequences $\lu^{(1)},\ldots,\lu^{(m)}$, it is decidable whether the zone description of the sequences with respect to $\mathbf{S_1},\ldots,\mathbf{S_k}$ belongs to $\mathcal L$. 
\end{theorem}

\section*{Acknowledgements}
   Shaull Almagor has received funding from the European Union's
   Horizon 2020 research and innovation programme under the Marie
   Sk{\l}odowska-Curie grant agreement No.~837327.
Jo\"el Ouaknine is supported by ERC grant AVS-ISS (648701) and DFG grant 389792660 as part of
TRR 248 (see \texttt{https://perspicuous-computing.science}).
James Worrell is supported by EPSRC Fellowship EP/N008197/1.

\bibliographystyle{apalike}
\bibliography{bibliography}

\appendix
\section{First-order theory of real closed fields}
\label{sec:foth}
In the first-order logic of real closed fields the atomic formulas are of the form:
\begin{align*}
  p(x_1,\ldots,x_k)\sim 0,
\end{align*}
where $p$ is a polynomial in $\intg[x_1,\ldots,x_k]$ and $\sim\in\set{>,=}$. In this logic  we are allowed to quantify over real numbers, and use the Boolean connectives. Subsets of $\rel^k$ that satisfy a formula $\Phi(x_1,\ldots,x_k)$ of this logic are called \emph{semi-algebraic sets.} The first-order logic of real closed fields admits effective quantifier elimination, a fact known as the Tarski's theorem, which we state as follows.
\begin{theorem}[{\cite[Theorem 37]{tarski1951decision}}]
  \label{th:tarski}
  There is an algorithm that inputs a sentence $\Phi$ (a formula without free variables) from the language above, and returns yes if and only if $\Phi$ is true over the real numbers. 
\end{theorem}

There is a natural first-order interpretation of the field of complex numbers into the field of real numbers, where every complex variable $z=x+\ii y$ is replaced by two real variables $x$ and $y$.

\begin{example}
  Consider the polynomial $p(z):=z^3+5z$. Its roots are $\lambda_1:=0$, $\lambda_2:=\ii\sqrt 5$, and $\lambda_3:=-\ii\sqrt 5$. The algebraic number $\lambda_2$ can be identified with the formula (to be interpreted over $\com$) $\phi_2(z)$ which says $p(z)=0$ and $\mathrm{Im}(z)>2$. Using the interpretation alluded above, we replace $\phi_2(z)$ with $\phi_2'(x,y)$ (where $x$ and $y$ range over reals) which says $p_1(x,y)=p_2(x,y)=0$ and $y>2$ where
  \begin{align*}
    p_1(x,y)\defeq x^3-3xy^2+5x,\qquad p_2(x,y)\defeq 3x^2y-y^3+5y.
  \end{align*}
  Clearly
  \begin{align*}
    \set{z\in\com\st p(z)=0\text{ and }\mathrm{Im}(z)>2}=\set{x+\ii y\in\com\st p_1(x,y)=p_2(x,y)=0\text{ and }y>2}. 
  \end{align*}
\end{example}

The intervals where the normalized roots of the characteristic polynomial associated to a \lrs lay, can be computed, therefore we assume that $\lambda=(\lambda_1,\ldots,\lambda_d)$ are given by formulas as in the example~above. One can take products and sums of such numbers, \ie compute a different formula which defines the product or sum. We will prove that the set $\torus_\lambda$ (defined in \Cref{subsec:walks}) is semi-algebraic. After this, we will give a full proof of \Cref{lem:fodef}.

\begin{lemma}
  \label{lem:torus semialgebraic}
  The set
  \begin{align*}
    \left\{(x_1,y_1,\ldots, x_d,y_d)\in\rel^{2d}\st (x_1+\ii y_1,\ldots,x_d+\ii y_d)\in\torus_\lambda \right\}
  \end{align*}
  is semi-algebraic.
\end{lemma}
\begin{proof}
  The set of multiplicative relations of $\lambda=(\lambda_1,\ldots,\lambda_d)$ is:
  \begin{align*}
    \mathcal{M}_\lambda\defeq\set{\mathbf{v}\in\intg^d\st\lambda_1^{v_1}\lambda_2^{v_2}\cdots \lambda_d^{v_d}=1}. 
  \end{align*}
  This is an Abelian subgroup of $(\intg^d,+)$, hence it has a finite basis $B\subset \intg^d$. Masser gave an explicit upper bound on the components of this basis in \cite[Section 4]{mas88}. As a consequence, $B$ is computable: to test whether some $(b_1,\ldots,b_d)$ is in $B$, or equivalently whether $\lambda_1^{b_1}\cdots \lambda_d^{b_d}=1$, use \Cref{th:tarski}. We recall the definition of $\torus_\lambda$:
\begin{align*}
  \torus_\lambda\defeq\set{\mathbf{z}\in\torus^d\st z_1^{v_1}z_2^{v_2}\cdots z_d^{v_d}=1\text{ for all }\mathbf{v}\in\mathcal{M}_\lambda}.
\end{align*}
Since $B$ is a basis of $\mathcal{M}_\lambda$, we can replace $\mathcal{M}_\lambda$ by $B$ in the definition above. So $\torus_\lambda$ is the set of all $(x_1,y_1,\ldots,x_d,y_d)$ such that for every $(b_1,\ldots,b_d)\in B$ we have:
\begin{align*}
  (x_1+\ii y_1)^{b_1}(x_2+\ii y_2)^{b_2}\cdots (x_d+\ii y_d)^{b_d}-1=0.
\end{align*}
Since this is a finite set of equations, the lemma follows. 
\end{proof}

We give a full proof of \Cref{lem:fodef}.

\fodeflem*
\begin{proof}
  Let $P,m\in\nat$ and $w=w(1)w(2)\cdots w(m)\in\set{-,0,+}^*$ be such that $w(i)$ are factors of length $P$.
  We recall the definitions from \Cref{subsec:proof of simple almost periodic}.
  \begin{align*}
  U(w)\defeq\set{\mathbf{x}\in\torus_\lambda\st g(\mathbf{x})=w(1), g\left(s(\mathbf{x})\right)=w(2),\ldots, g\left(s^{m-1}(\mathbf{x})\right)=w(m)},
\end{align*}
where $g\st \torus_\lambda\to \set{-,0,+}^P$ is the composition of $f$ and $\sgn$ applied component-wise, $f$ is the linear map $f\st \torus_\lambda\to\rel^P$ that we get after applying \Cref{lem:same sign} to every subsequence $\lu_\ell$, and finally $s$ maps $(x_1,\ldots,x_d)$ to $(\lambda_1x_1,\ldots,\lambda_dx_d)$. Since we have formulas for $\lambda_i$, we can compute formulas for $s^k(\mathbf x)$, for any $k\in\nat$. Inspecting the proof of \Cref{lem:same sign}, reveals that the constants $z_1,\ldots, z_d$ have computable formulas, hence the same holds for $f(s^k(\mathbf x))$, \ie we can compute $\varphi_k(\mathbf x, \mathbf y)$, such that 
\begin{align*}
  \varphi_k(x_1,\ldots,x_d,y_1,\ldots,y_P)\Leftrightarrow \mathbf y = f(s^k(\mathbf x)). 
\end{align*}
To require that $g(s^k(\mathbf x))=v$ for some $v\in\set{-,0,+}^P$, we use the formula $\varphi_k$ and ask that $y_i\sim 0$ where $\sim\in\set{<,=,>}$ depending on whether $v_i$ is ``$-$'', ``$0$'', or ``$+$''. Since the set of $\mathbf x\in\torus_\lambda$ is semi-algebraic thanks to \Cref{lem:torus semialgebraic}, and we have constructed formulas that require $g(s^k(\mathbf x))=v$, the lemma is proved. 
\end{proof}

\end{document}